\documentclass[10pt,journal]{IEEEtran} 

\usepackage[linesnumbered,ruled,vlined]{algorithm2e}
\usepackage{mathtools, cases}
\usepackage{textcomp}
\usepackage{makecell}
\usepackage{amsthm, enumitem}
\usepackage{amsmath,amsfonts,amssymb}
\usepackage{cite, url}
\usepackage{verbatim}
\usepackage{bm}
\usepackage{graphicx}
\usepackage{xcolor}
\usepackage{subcaption}
\usepackage{epstopdf}
\usepackage{mathrsfs}
\usepackage{txfonts}
\usepackage{lineno}
\usepackage{multirow}
\usepackage{booktabs}
\usepackage{color}
\usepackage{multicol}
\usepackage{stfloats}
\usepackage{diagbox}
\usepackage{esint}
\usepackage{float}
\usepackage{algpseudocode}
\usepackage{empheq}
\usepackage{tabularx}

\usepackage{flushend}

\newtheorem{lemma}{Lemma}
\allowdisplaybreaks[4]

\newtheorem{remark}{Remark}

\begin{document}

\title{Latency-Aware Generative Semantic Communications with Pre-Trained Diffusion Models}

\author{Li Qiao, Mahdi Boloursaz Mashhadi,~\IEEEmembership{Senior Member,~IEEE}, Zhen Gao,~\IEEEmembership{Member,~IEEE},\\ Chuan Heng Foh,~\IEEEmembership{Senior Member,~IEEE}, Pei Xiao,~\IEEEmembership{Senior Member,~IEEE}, and Mehdi Bennis,~\IEEEmembership{Fellow,~IEEE}

 \vspace{-10mm}
 

\thanks{This work was supported in part by the Shandong Province Natural Science Foundation under Grant ZR2022YQ62, in part by the Beijing Nova Program, in part by the National Natural Science Foundation of China (NSFC) under Grant U2233216 and Grant 62071044, and in part by the U.K. EPSRC  under Grant  EP/X013162/1.~~~L. Qiao is with the School of Information
and Electronics, Beijing Institute of Technology, Beijing 100081, China, also with 5GIC \& 6GIC, the Institute for Communication Systems (ICS), University of Surrey, GU2 7XH Guildford, U.K. (e-mail: qiaoli@bit.edu.cn).~~M. Boloursaz Mashhadi, C. H. Foh, and P. Xiao are with 5GIC \& 6GIC, Institute for Communication Systems (ICS), University of Surrey, GU2 7XH Guildford, U.K. (emails: \{m.boloursazmashhadi, c.foh, p.xiao\}@surrey.ac.uk).~~Z. Gao is with State Key Laboratory of CNS/ATM, Beijing Institute of Technology (BIT), Beijing 100081, China (e-mail: gaozhen16@bit.edu.cn).~~M. Bennis is with the Centre for Wireless Communications, University of Oulu, 90014 Oulu, Finland (e-mail: mehdi.bennis@oulu.fi).}

}

\maketitle

\begin{abstract}

Generative foundation AI models have recently shown great success in synthesizing natural signals with high perceptual quality using only textual prompts and conditioning signals to guide the generation process. This enables semantic communications at extremely low data rates in future wireless networks. In this paper, we develop a latency-aware semantic communications framework with pre-trained generative models. The transmitter performs multi-modal semantic decomposition on the input signal and transmits each semantic stream with the appropriate coding and communication schemes based on the intent. For the prompt, we adopt a \textit{re-transmission}-based scheme to ensure reliable transmission, and for the other semantic modalities we use an \textit{adaptive modulation/coding} scheme to achieve robustness to the changing wireless channel. Furthermore, we design a semantic and latency-aware scheme to allocate transmission power to different semantic modalities based on their importance subjected to semantic quality constraints. At the receiver, a pre-trained generative model synthesizes a high fidelity signal using the received multi-stream semantics. Simulation results demonstrate ultra-low-rate, low-latency, and channel-adaptive semantic communications.

\end{abstract}
\vspace{-2mm}

\IEEEpeerreviewmaketitle

\vspace{-3mm}
\section{Introduction}
\textit{Semantic Communication (SemCom)} is poised to play a pivotal role in shaping the landscape of future AI/ML-driven communication systems. SemCom systems aim to significantly reduce the transmission rate by extracting and transmitting only the semantic message of interest based on the communication intent. \textit{Generative AI (GenAI)} models have recently proved to significantly enhance communication at the semantic-level \cite{xia2023generative, grassucci2023generative, li2024extreme}. GenAI models such as Diffusion \cite{StableDiffusion,NEURIPS2021_49ad23d1}, Flow-based \cite{Flow}, and GAN \cite{NIPS2014_5ca3e9b1} models, can learn the general distribution of natural signals through training and can generate new samples at the inference time. This generative process can be guided or conditioned to synthesize outputs with a desired semantic content. In \textit{Generative Semantic Communications (Gen SemCom)}, the semantics of interest are extracted at the transmitter, communicated over the channel, and then used at the receiver to guide a generative model to synthesize a semantically consistent signal with high fidelity. GenAI models are trained to maximize the perceptual quality and the fundamental bounds on Generative SemCom are governed by the \textit{rate-distortion-perception} theory \cite{RDP2, RDP1}, which determines the threefold trade-off between rate, distortion, and perceptual quality of the reconstructed signal.

The recent advent of powerful \textit{Generative Foundation Models} provides ample opportunities to develop ultra-low-rate semantic communication systems. The ultra low rate transmission can be achieved by transmitting data semantics in compressed format as a textual message or \textit{prompt}. For instance, the prompt \textit{“Teddy bear surfer rides the wave in the tropics”} can be used to generate a short video with its semantic content matching with the prompt. The generative foundation models such as Sora \cite{Sora}, Lumiere \cite{bar2024lumiere}, and DALL.E \cite{DallE} are pre-trained on large amount of data and can synthesize various types of AI Generated Content (AIGC) with high quality. The pre-trained nature of such models and their applicability to a vast range of multi-media synthesis tasks, has the potential to revolutionize generative semantic communications enabling universal intent and channel-adaptive SemCom systems empowered by pre-trained foundation models.

In this paper, we develop a universal generative semantic communications framework with pre-trained foundation models. We claim two key benefits in comparison with the existing semantic communication frameworks. Firstly, the foundation models possess a vast general knowledge leveraging the intensive self-supervised training process on huge amount of data. This alleviates the need for a shared knowledge base/graph between the semantic transmitter and receiver, obviating the need for corresponding knowledge sharing overheads imposed in current SemCom frameworks \cite{SemCom1, Semcom2, Semcom3, Semcom4, Semcom5}. This vast general knowledge also makes the proposed Generative SemCom framework applicable to various datasets and tasks thereby achieving universality. Secondly, the adoption of pre-trained models allows a separation-based SemCom architecture, alleviating the need for end-to-end joint training of the transmitter and receiver, which is required in many SemCom frameworks \cite{SemCom1, Semcom2, Semcom3, Semcom4}. Such a separation-based architecture offers better compatibility with the existing design of wireless communication networks in comparison with the end-to-end methods. The proposed framework is specifically suitable for scenarios which require communication of huge multi-modal data with stringent latency and reliability requirements, e.g. the wireless metaverse, extended/mixed reality (XR/MR), holographic teleportation, and the internet of senses.
The contributions of this work are three-fold:
\begin{itemize}
    \item We develop a semantic decomposition scheme at the transmitter which extracts the semantic content of the input signal in multiple semantic modalities. We extract the most important semantic contents as a compact textual message or prompt, along with multiple other modalities that act as conditioning signals to guide the synthesis process at the generative foundation model at the receiver.
    \item To achieve semantic-aware communication, we design a multi-stream scheme that transmits each extracted semantic modality with appropriate coding and communication techniques based on communication intent. Due to the importance of the textual prompt, a re-transmission scheme is applied to ensure reliable reception, while other modalities are transmitted with a modulation scheme adapted to the varying wireless channel.
    \item We design a semantic and latency-aware scheme to allocate the transmission power to different semantic streams based on semantic importance, and to adapt the modulation order to the varying wireless channels.
\end{itemize}




\textit{\textbf{Notations}}: Boldface lower and upper-case symbols denote column vectors and matrices, respectively. $[K]$ denotes $\{1,2,...,K\}$. $[{\bf v}]_n$ and $\left| {\bf v} \right|$ denote the $n$-th element and the length of vector ${\bf v}$, respectively. Finally, ${\bf v}\cdot {\bf v}$ is the dot production, and $Pr\{\cdot\}$ denotes probability
of an event.



\vspace{-4.5mm}
\begin{figure}[t]
		\centering  		
  \includegraphics[width=0.48\textwidth]{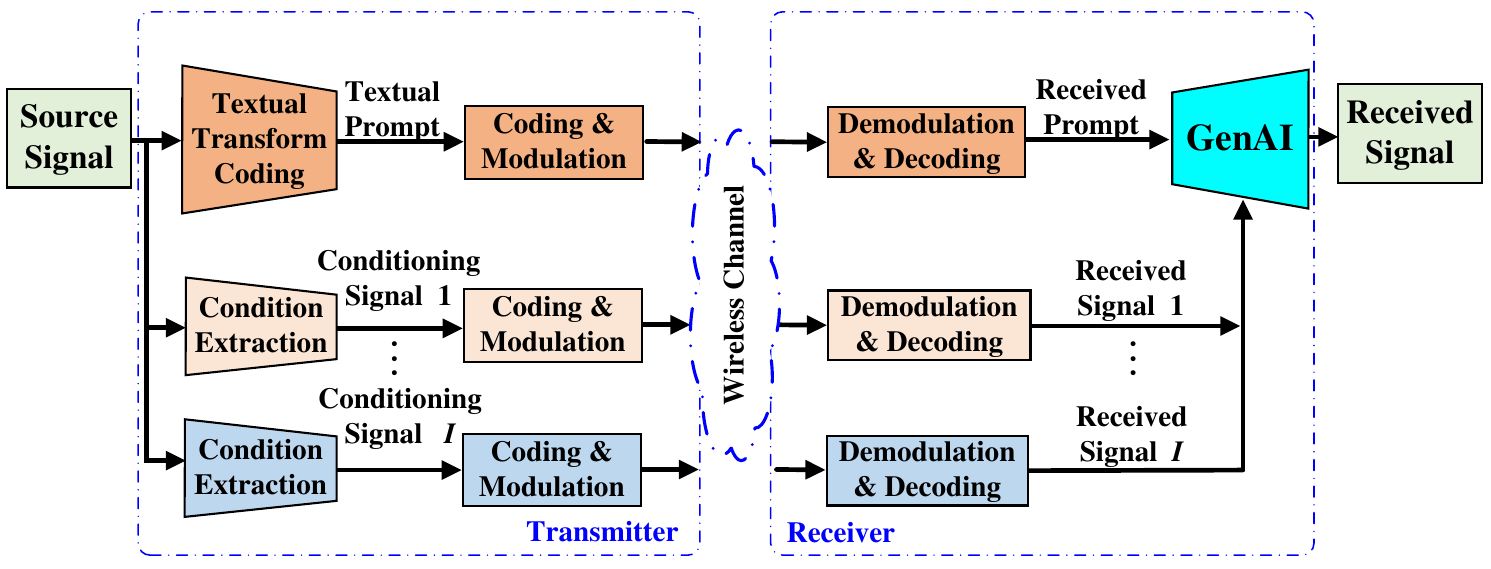}
  \captionsetup{font={footnotesize}, singlelinecheck = off, justification = raggedright,name={Fig.},labelsep=period}
		\caption{The Proposed Framework for Latency-aware Multi-stream Semantic Communication with Multi-Modal Generative Models.}
  \vspace{-7mm}
		\label{SystemDiag}
    \end{figure}

\section{Generative Foundation AI-based Semantic Communication}
Fig. \ref{SystemDiag} depicts our proposed framework for generative semantic communication with pre-trained foundation models. This framework includes multi-modal semantic decomposition/synthesis, semantic-aware multi-stream transmission, and latency-aware semantic power allocation.

\vspace{-3.7mm}
\subsection{Multi-modal Semantic Decomposition and Synthesis}

At the transmitter, a pre-trained textual transform encoder performs ultra-low-rate source-to-text transformation, extracting the textual message which acts as a prompt for the GenAI process at the receiver. Depending on the source signal (e.g. voice, image, video, point cloud, etc.), various large, medium, and small scale pre-trained AI models can be used for prompt extraction. The prompt contains the most important semantic contents of the source signal. At the same time, additional semantic information is extracted from the source signal based on communication intent. They are typically extracted in modalities other than text and provide additional conditioning signals to guide the GenAI process at the receiver. For example, if the source signal to be transmitted is a short video, the initial frame of that video can be a good conditioning signal to be transmitted along with the prompt to guide the generation process at the receiver. As another example, if the source signal is an image and the receiver is interested in the relative location of different objects in the image, an edge map can be a good conditioning signal, as discussed later in Section \ref{PerformanceAnal}. 

The decomposed multi-modal semantics, i.e., prompt and conditioning signals, are then compressed, coded, and modulated, and transmitted over the wireless channel. At the receiver, after demodulation, decoding, and decompression, the received multi-modal semantics are simultaneously fed into a pre-trained generative model for high-fidelity synthesis of the source signal. We note that state-of-the-art generative models, e.g. stable diffusion (SD) \cite{StableDiffusion}, allow generation using only the prompt, which results in a relatively high distortion. On the other hand, some semantic information content of the source signal can be more efficiently communicated using other modalities. The conditioning signals improve the semantic fidelity as well as distortion at additional communication costs. Finally, the desired dimensions of the signal to be synthesized by GenAI are assumed to be fixed and known to the receiver.

 \vspace{-4.5mm}
\subsection{Semantic-aware Multi-Stream Transmission}

The transmitter sends the extracted multi-modal semantics, i.e., prompt and conditioning signals, in multiple streams over orthogonal frequency channels to enable semantic-aware design of specific transmission mechanisms, e.g. coding rate, modulation order, for different semantic modalities. Since the prompt is a compact textual message and contains the most important semantic content of the source signal, we use lossless compression for prompt to obtain a bit-stream ${\bf v}_0$. However, for other semantic modalities we use state-of-the-art Deep Neural Network (DNN)-based techniques for lossy compression as will be discussed in Section III. The bit-stream of the $i$-th conditioning signal is denoted by ${\bf v}_i$ , $i\in[I]$. 

We consider the Rayleigh fading channel and assume that the channels experience block fading within the transmit duration of each semantic modality. In particular, the channel gain of the $i$-th semantic modality is expressed as
\begin{align}\label{eq-channel}
h_{i} = \sqrt{\epsilon_o \left(d \right)^{-\varphi}} \tilde{h},~\forall i\in\{0\} \cup[I],
\end{align}
where $i=0$ (resp. $i\in [I]$) denotes the prompt (resp. conditioning signals), 
$\tilde{h}$ is a random scattering element captured by zero-mean and unit-variance circularly symmetric complex Gaussian (CSCG) variables, $\epsilon_o$ is the path loss at the reference distance $d_0 =1$ m, $\varphi$ is the path loss exponent, $d$ is the distance. Furthermore, the received signal-to-noise ratio (SNR) at the receiver for the $i$-th semantic modality is given by $\gamma_{i}(p_i) = \frac{ p_{i}|h_{i}|^2}{B_i N_0},~\forall i\in\{0\} \cup[I]$, 
where $p_{i}$ and $B_i$ denote the transmit power and the bandwidth allocated for the $i$-th modality, respectively, and $N_0$ is the noise power spectral density. For notational simplicity, we use $\gamma_i$ instead of $\gamma_i(p_i)$, $\forall i\in\{0\} \cup[I]$, in the sequel.

\subsubsection{\color{black}Re-transmission-based Communication of the Textual Prompt}
The textual prompt is a super compact message that is highly sensitive to errors. Even a single bit error can change a character/word causing a significant semantic error. Thereby, our framework uses re-transmissions if the prompt is received with any error. Based on this, we propose a cyclic redundancy check (CRC)-based re-transmission mechanism for the prompt to guarantee its reliable communication. Assuming the prompt is transmitted in packets of length $L$ bits, the packet error rate (PER) can be expressed as \cite{wu2014ARQ}
\begin{equation}\label{eq-per}
\text{PER}(p_0) \approx 1 -  L^{-k/\gamma_0} \exp\left( -\dfrac{b}{\gamma_0} \right),
\end{equation}
where $k$ and $b$ are parameters determined by the channel coding scheme and the modulation order. We denote the coding rate and the modulation order as $r$ and $M_{0}$, respectively. Hence, the expected total transmission delay for the prompt is given by
\begin{equation}\label{eq-delay1}
{T_0}(p_0) = \eta_{P} \cdot \eta_{R} \cdot \eta_{T}  = \dfrac{{\left| {{{\bf{v}}_0}} \right|}}{{{r}L}} \cdot \dfrac{1}{{1 - {\text{PER}(p_0)}}} \cdot \dfrac{L}{{{\log _2}(M_{0}){B_0}}},
\end{equation}
where $\eta_{P}=\frac{{\left| {{{\bf{v}}_0}} \right|}}{{{r}L}}$, $\eta_{R}=\frac{1}{{1 - {\text{PER}(p_0)}}}$, $ \eta_{T}=\frac{L}{{{\log _2}(M_{0}){B_0}}}$ denote the number of packets, the number of re-transmissions, and the transmission delay of each packet, respectively.

\subsubsection{Adaptive Modulation and Coding for other Semantic Modalities} 
For the transmission of the other semantic modalities, we adopt an adaptive MQAM modulation and coding scheme to cope with varying wireless channel. Denote the bit error rate (BER) of the $i$-th, $i\in[I]$, conditioning signal as $\text{BER}_i=\frac{1}{|{\bf{v}}_i|}\sum\nolimits_{n=1}^{|{\bf{v}}_i|}Pr\{[\hat{\bf{v}}_i]_n \neq [{\bf{v}}_i]_n\}$, where $\hat{\bf{v}}_i$ denotes the corresponding received data stream of the conditioning signal. With such a scheme, the achievable rate can be generally expressed as \cite{goldsmith1997variable}
\begin{equation}\label{eq-R-ber}
R_i(p_i, \text{BER}_i) = B_i{\log _2}\left( 1 + \dfrac{\alpha}{-\ln\left(\beta \text{BER}_i\right)} \gamma_i  \right), i\in[I],
\end{equation}
where $\alpha$ and $\beta$ are parameters determined by the coding schemes. Thereby, the transmission delay of $i$-th semantic stream is ${T_i}(p_i,\text{BER}_i) = \frac{{\left| {{{\bf{v}}_i}} \right|}}{{{R_i}(p_i,\text{BER}_i)}}$, $\forall i\in[I]$.

    \vspace{-3.8mm}
\subsection{Latency-aware Adaptive Semantic Communication}
In this subsection, we propose a semantic and latency-aware approach for selecting the optimal communication parameters such as the transmission power, and modulation order, for various semantic modalities, tailored to the prevailing channel quality as well as the specified semantic quality requirements. We formulate the problem as minimizing the transmission delay under the power and semantic quality constraints given as 
\begin{equation}\label{OptiProb}
\begin{aligned}
& \min_{\left\{p_0,p_{i},i\in[I]\right\}} & & \{\max\left\{T_0(p_0), T_{i}(p_i, \text{BER}_i),~i\in[I]~\right\} \}\\
& \text{s.t.} & &  \sum\nolimits_{i=0}^I p_{i}\leq P_{\text{T}},\\
&&& \Phi_j\left(\text{BER}_i,~i\in[I]\right) \geq \varepsilon_j,~j\in[J], \\
\end{aligned}
\end{equation}
where $P_{\text{T}}$ is the maximum transmit power budget, $\Phi_j$ and $\epsilon_j$ denote the $j$-th semantic quality metric and its requirement\footnote{It is assumed that a larger value of the metric represents improved semantic quality. Many common metrics can be easily transformed to satisfy this.}, respectively. Since reliable transmission of the prompt is assumed via re-transmissions, $\Phi_j$, $\forall j\in[J]$, is a multi-variable function of the BERs of the $I$ conditioning signals, denoted as $\Phi_j\left(\text{BER}_i,~i\in[I]\right)$. The quality requirements can represent any semantic distortion or perception \cite{RDP2, RDP1} requirement of the signal synthesized at the receiver, and can be enforced by various metrics, e.g. MS-SSIM \cite{SSIM}, LPIPS \cite{LPIPS}, FID \cite{FID}, CLIP \cite{radford2021learning}, etc. The semantic quality constraints represent epigraphs of the metrics, and in the most general case, convexity of problem (\ref{OptiProb}) depends on the behaviour of the selected metrics in terms of $\text{BER}_i,~i\in[I]$. In the typical scenario, where the prompt is transmitted along with only one additional modality,  (\ref{OptiProb}) is convex with a straightforward solution as given in \textit{\textbf{Lemma 1}}.
\begin{lemma}
    For the case of one conditioning signal, i.e., $I=1$, if $\Phi_j(\text{BER}_1)$, $\forall j\in[J]$, are monotonically non-increasing with respect to $\text{BER}_1$, (\ref{OptiProb}) is a convex problem. The optimal solution can be achieved if and only if $p_0+ p_{1}= P_{\rm{T}}$, $T_0(p_0)=T_1(p_1, \text{BER}_1)$, and $\text{BER}_1=\min\left\{\Phi_j^{-1}(\varepsilon_1),~j\in[J]\right\}$, where $\Phi_j^{-1}$ denotes the generalized inverse function of $\Phi_j$, $\forall j\in[J]$.
\end{lemma}

\begin{proof}
    The epigraph of a single-variable function is convex iff the function is (weakly) monotonic. Hence, if $\Phi_j(\text{BER}_1)$, $\forall j\in[J]$, are monotonically non-increasing, then the intersection of their epigraphs determines a convex region with respect to $\text{BER}_1$, given by $\text{BER}_1 \le \min\left\{\Phi_j^{-1}(\varepsilon_1),~j\in[J]\right\}$ where $\Phi_j^{-1}$ exists because $\Phi_j$ is (weakly) monotonic. Moreover, since $\text{BER}_1$ is a decreasing function of $p_1$ itself, the constraints collectively determine a convex feasible region in terms of $p_0, p_1$. Using (\ref{eq-delay1}) and (\ref{eq-R-ber}), the latency functions $T_0, T_1$ are also convex in $p_0, p_1$, and thereby the optimization is convex with a global solution. Finally, it is straightforward to show that in the optimum solution, all constraints hold with equality and we have $T_0=T_1$.  Hence, the optimum solution can be found by solving the non-linear system of equations $T_0(p_0)=T_1(p_1, \text{BER}_1)$, $p_0+ p_{1}= P_{\rm{T}}$, and $\text{BER}_1=\min\left\{\Phi_j^{-1}(\varepsilon_1),~j\in[J]\right\}$ in $p_0, p_1$. The overall transmission delay $T$ is then $T=T_0(p_0)=T_1(p_1, \text{BER}_1)$.
\end{proof}

\vspace{-5mm}
\section{Simulation Results}\label{PerformanceAnal}

\subsection{System setup}
We consider a generative image semantic communication setting, where the receiver is interested in the main objects inside the image, e.g. ``car", ``building", and its general structure, i.e. object shapes and locations. In this context, we consider a two-modal semantic decomposition framework, including a textual prompt to convey the main objects, and an edge map as the conditioning signal to convey the general image structure. To extract the prompt, we adopt the GPT-4 model \cite{achiam2023gpt} in this work. We also adopt the holistically-nested edge detection model \cite{xie2015holistically}, to extract a gray-scale edge map from the original image at the transmitter. The prompt is encoded to bits using UFT-8 \cite{yergeau2003utf}, and the edge map is compressed to bits using learned nonlinear transform coding \cite{balle2020nonlinear}. At the receiver, we adopt the state-of-the-art pre-trained SD model \cite{StableDiffusion} for conditional image generation from the received prompt and the edge map. The wireless transmission parameters for the two semantic streams are listed in Table \ref{tab:wireless}. 


 \begin{table}[t]
 \captionsetup{font=small}
	\caption{Parameter settings}
 \vspace{-4.2mm}
 \scriptsize
	\begin{center}
		\begin{tabular}{ll}
			\toprule[1.5pt]
			Parameters  & Values \\ \hline
   Packet length of prompt, $L$    & $\left| {{{\bf{v}}_0}} \right|$ bits \\
   Coding rate of prompt, $r$    & 1/2 (convolutional code) \\
   Modulation order of prompt, $M_0$    & 4 (QPSK) \\
    $\{k, b\}$ in (\ref{eq-per})  & $\{0.374, -0.31\}$\\
    $\{\alpha, \beta\}$ in (\ref{eq-R-ber})  & $\{1.5, 5\}$\\
			Transmit power budget, $P_{\text{T}}$    & 10 mW \\ 
			Path loss at $d_0 =1$ m, $\epsilon_o$      &-30 dB \\ 
			Path loss exponent, $\varphi$   & 3.4 \\ 
			Noise power density, $N_0$                 & -174 dBm/Hz \\ 
			Bandwidth, $B_0=B_1$                         & 1~MHz \\
			\toprule[1.5pt]
		\end{tabular}
	\end{center}
	\label{tab:wireless}
\vspace{-4.5mm}
\end{table}

 \begin{figure}[t]
  \centering
  \vspace{-3mm}
  \captionsetup[subfigure]{font=scriptsize, labelfont=scriptsize, skip=0.5pt, position=bottom} 
  \begin{subfigure}[b]{0.46\columnwidth}
    \centering
\includegraphics[width=\linewidth]{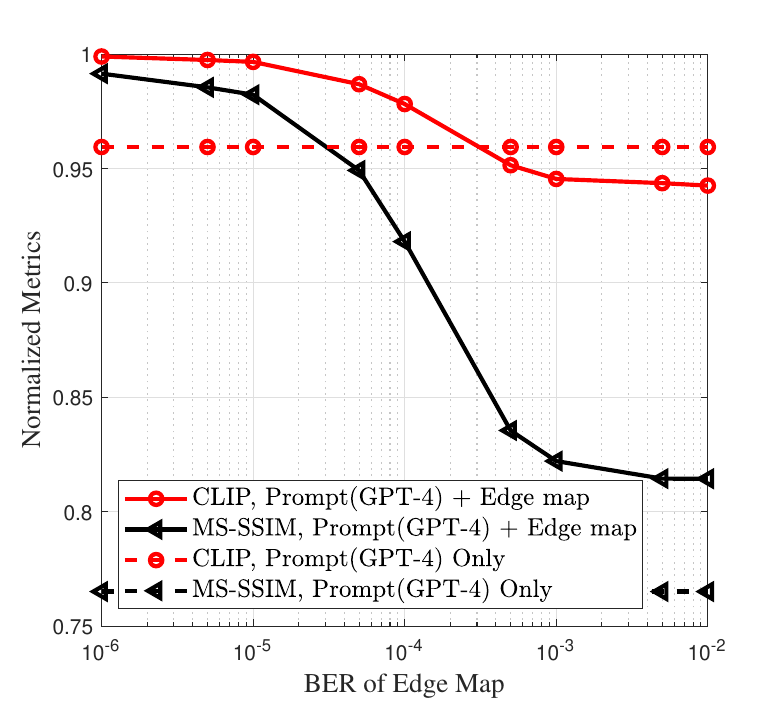}
    \caption{}
    \label{fig:sub1}
  \end{subfigure}%
  \begin{subfigure}[b]{0.46\columnwidth}
    \centering
    \includegraphics[width=\linewidth]{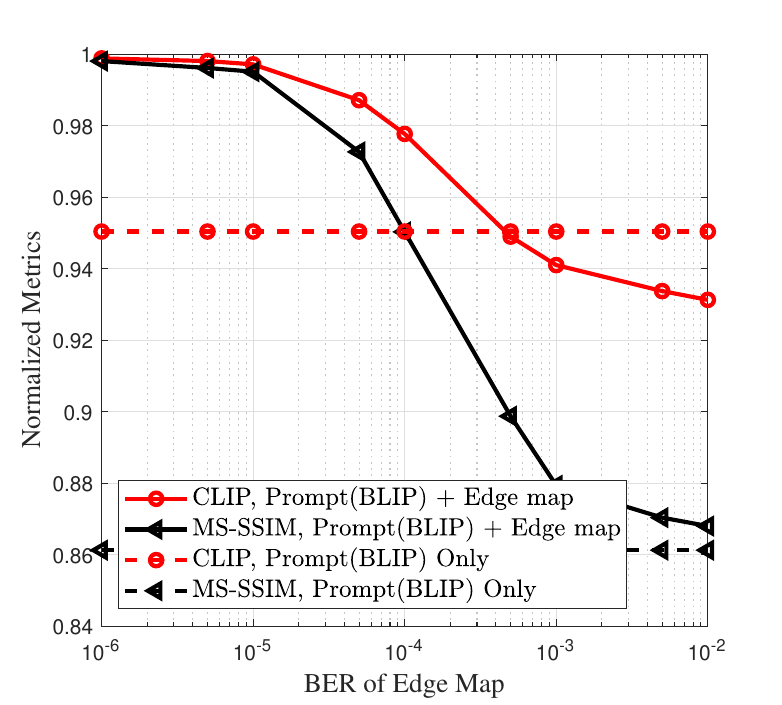}
    \caption{}
    \label{fig:sub2}
  \end{subfigure}
  \vspace{-2mm}
  \captionsetup{font={footnotesize}, singlelinecheck = off, justification = raggedright,name={Fig.},labelsep=period}
  \caption{Normalized CLIP and MS-SSIM versus BER of the edge map: (a) Prompt generated by GPT-4; (b) Prompt generated by BLIP.}
  \vspace{-7.9mm}
  \label{PandD_values}
\end{figure}

\vspace{-5mm}
\subsection{Semantic quality metrics}

Various reference-based or no-reference metrics can be adopted to evaluate the quality of the synthesized signal based on the semantic communication intent. In this work, we adopt two metrics, one to measure the semantic similarity, and the other to measure the structural similarity between the original and the synthesized images at the \textit{Gen SemCom} receiver. In order to assess the quality of the most important semantic information conveyed by the prompt, we use the contrastive language-image pretraining (CLIP) model \cite{radford2021learning}. Specifically, we define the CLIP metric as the cosine similarity between the CLIP embedding of the original and synthesized images, i.e. $f({\bf x})$ and $f(\hat{\bf x})$, given by $\text{CLIP} = ( \frac{f({\bf x})\cdot f(\hat{\bf x})}{\| f({\bf x}) \|_2 \cdot \| f(\hat{\bf x}) \|_2} + 1)/2$. For structural similarity, we use the multi-scale structural similarity (MS-SSIM) \cite{SSIM} metric.

\vspace{-4mm}
\subsection{Monotonicity of the quality metrics} \label{Monotone}
The end-to-end performance of the proposed framework depends on behaviour of the pre-trained GenAI model at the receiver, i.e. its sensitivity to the BER of the conditioning signal, and the requirements of the semantic quality metrics. Hence, we investigate monotonicity of the CLIP and MS-SSIM metrics to apply \textit{\textbf{Lemma 1}} for latency-aware adaptive SemCom. Specifically, through intensive simulations in presence of random bit errors in the proposed framework, we demonstrate that the normalized CLIP/MS-SSIM metrics are monotonically non-increasing functions of the BER of the edge map, as depicted in Fig. \ref{PandD_values}. Note that the proposed framework works independent of the choice of the AI models adopted at the transmitter and receiver. For example, the GPT-4 model can be replaced by smaller image-to-text models, e.g. BLIP \cite{li2022blip}, Oscar \cite{li2020oscar}, UNIMO \cite{li2020unimo}, at the cost of a slight degradation of the semantic quality metrics. These smaller models can be implemented locally on-device. The absolute CLIP values achieved with GPT-4 and BLIP, under error-free transmission of the edge map, are $0.918$ and $0.896$, respectively, while the resulting MS-SSIM values are similar. The normalized CLIP and MS-SSIM are defined as $\frac{\text{CLIP}}{\text{CLIP}_0}$ and $\frac{\text{MS-SSIM}}{\text{MS-SSIM}_0}$, respectively, where CLIP$_0$ and MS-SSIM$_0$ are the reference values of these metrics for perfect transmission of the edge map. We note that ``CLIP, Prompt + Edge Map" surpasses ``CLIP, Prompt Only" when $\text{BER}\leq 10^{-4}$, indicating that transmitting the edge map enhances the semantic quality as interpreted by the SD model. On the contrary, when $\text{BER} > 10^{-3}$, the inaccuracies in the received edge map due to the communication errors become detrimental, causing the SD model to misinterpret the prompt semantics. For example in Fig. \ref{fig:SystemResults}, the ``fence" is missing in the generated image due to the inaccuracies of the received edge map at $\text{BER}=10^{-3}$, despite the word ``fence" being part of the prompt. Moreover, transmitting the edge map alongside the prompt invariably boosts the MS-SSIM in comparison to relying solely on prompt.

\vspace{-5mm}
\begin{figure}[t]
		\centering	\includegraphics[width=0.45\textwidth]{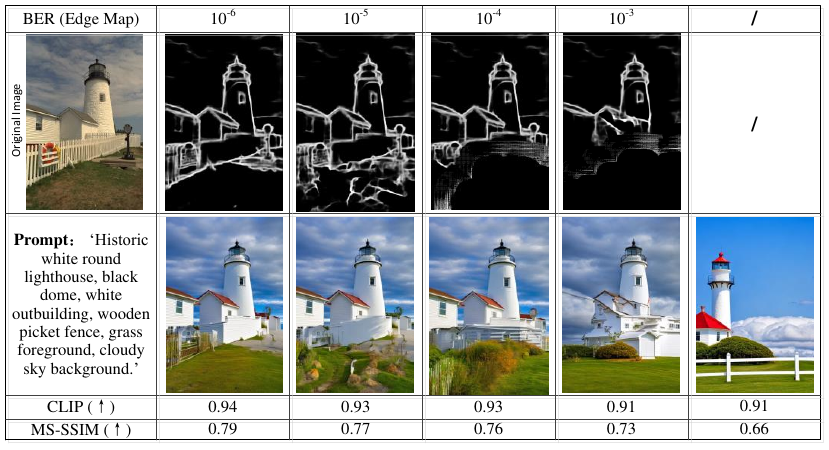}
  \captionsetup{font={footnotesize}, singlelinecheck = off, justification = raggedright,name={Fig.},labelsep=period}
		\caption{Visualization of the semantic quality of our proposed framework. The absolute CLIP/MS-SSIM values are reported in this figure for comparison.}
  \vspace{-7.5mm}
		\label{fig:SystemResults}
    \end{figure}
    
\subsection{Visual quality of the proposed framework}
In Fig. \ref{fig:SystemResults}, we provide illustrations to demonstrate the visual quality of our proposed framework on a sample natural image. These results demonstrate a good semantic quality for the proposed framework for ultra-low-rate transmission at bit per pixel (bpp) values as low as $0.0024$ and $0.017$ for the prompt and the edge map, respectively.  


\vspace{-4.5mm}
\subsection{Latency-aware Adaptive Semantic Communication}

We evaluate the optimal semantic latency in terms of the average SNR, i.e. $\overline{\gamma} = \frac{p_{\text{T}}\mathbb{E}(|h_{0}|^2)}{(B_0+B_1) N_0}$, as a channel quality indicator. With the investigation in subsection \ref{Monotone}, we apply \textit{\textbf{Lemma 1}} and solve the resulting nonlinear system of equations numerically. Fig. \ref{fig:four_subfigures} quantifies the semantic quality-latency trade-off for the proposed \textit{Gen SemCom} framework, and illustrates the optimal wireless parameters versus average SNR $\overline{\gamma}$ at various target semantic quality, i.e. normalized (CLIP, MS-SSIM), and BER values. It is evident that to achieve a higher target semantic quality (i.e. higher CLIP/MS-SSIM and lower BER) one should allocate more power to the transmission of the edge map, which in turn, increases the latency. Furthermore, when the channel SNR deteriorates, the sender should use a lower modulation order for the edge map to maintain the semantic quality at a target acceptable level. This will also increase the expected number of re-transmissions for the prompt.

For performance comparison, we adopt a semantic-unaware transmission benchmark, in which the prompt and edge map are treated equally and transmitted as one data stream regardless of the unequal importance of each semantic modality. For a fair comparison, the maximum power $p_{\rm T}$ and bandwidth $B_0+B_1$ are allocated for the transmission of this single stream with adaptive MQAM employed. The prompt BER becomes the bottleneck for single-stream transmission, and hence, the BER is kept at $10^{-7}$ to ensure reliable reception of the prompt. From Fig. \ref{fig:sub2}, the proposed semantic-aware multi-stream transmission significantly reduces the transmission latency at the lower SNR region, i.e., $\overline{\gamma}<12$\,dB. Finally, in Table \ref{tab:regions}, we provide modulation adaptation guidelines to achieve various semantic  qualities (CLIP, MS-SSIM) in varying channel qualities. Here, $M_1$ is the modulation order for the edge map, and $T$ is the transmission latency.

\vspace{-2mm}
    \begin{figure}[t]
  \centering
\captionsetup[subfigure]{font=scriptsize, labelfont=scriptsize, skip=0.5pt, position=bottom} 
  \vspace{-2mm}
  \begin{subfigure}[b]{0.44\columnwidth}
    \centering
\includegraphics[width=\linewidth]{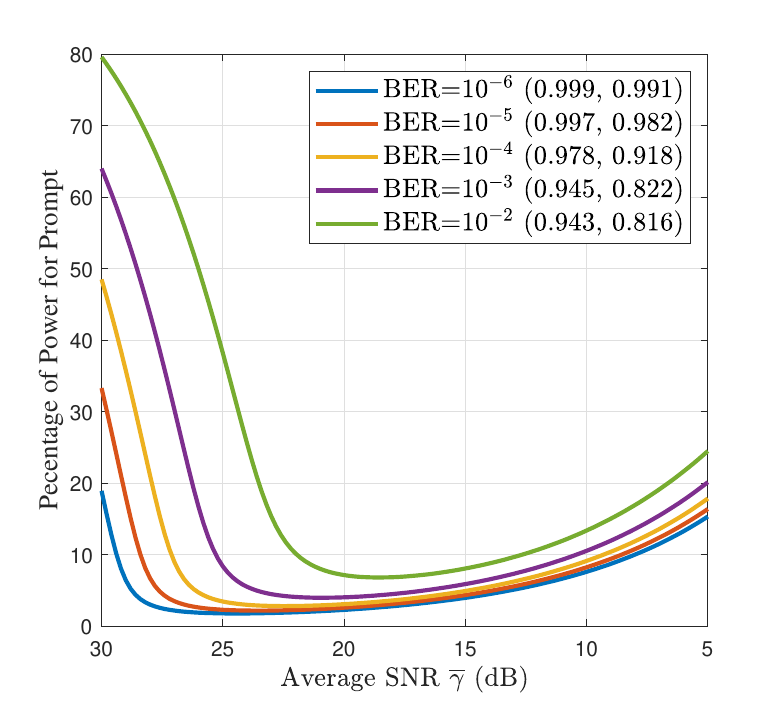}
    \caption{}
    \label{fig:sub1}
  \end{subfigure}%
  \begin{subfigure}[b]{0.44\columnwidth}
    \centering
    \includegraphics[width=\linewidth]{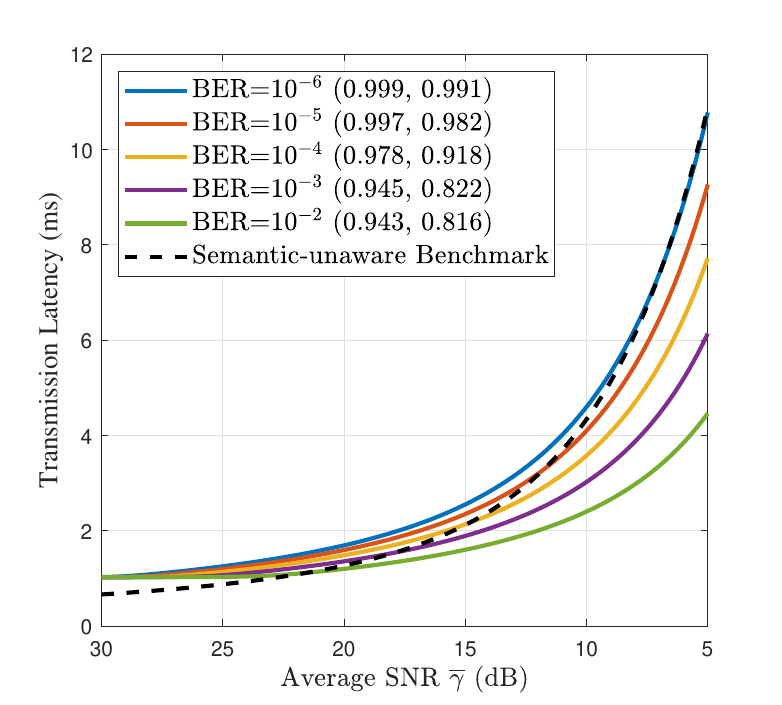}
    \caption{}
    \label{fig:sub2}
  \end{subfigure}
  \vspace{-2mm}
  \begin{subfigure}[b]{0.44\columnwidth}
    \centering
    \includegraphics[width=\linewidth]{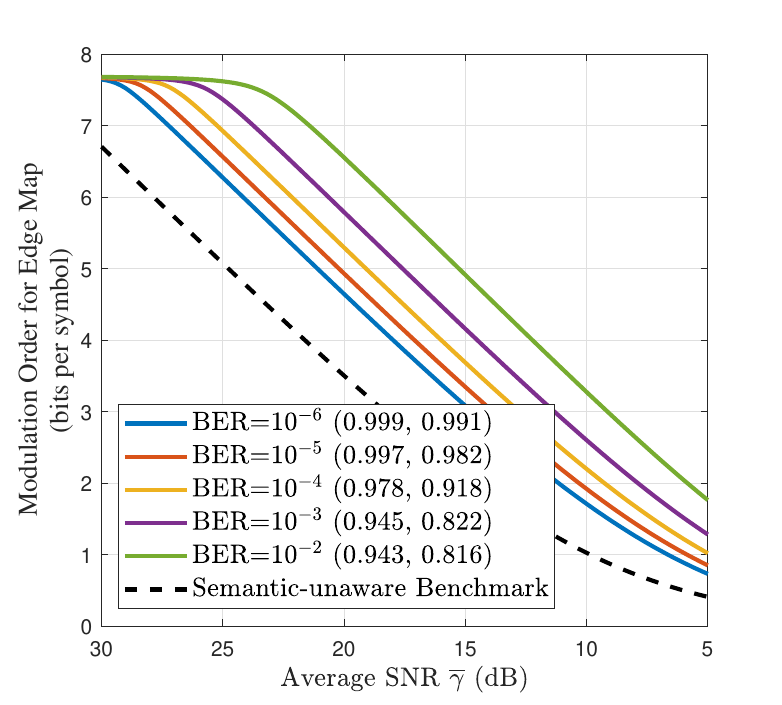}
    \caption{}
    \label{fig:sub3}
  \end{subfigure}%
  \begin{subfigure}[b]{0.44\columnwidth}
    \centering
    \includegraphics[width=\linewidth]{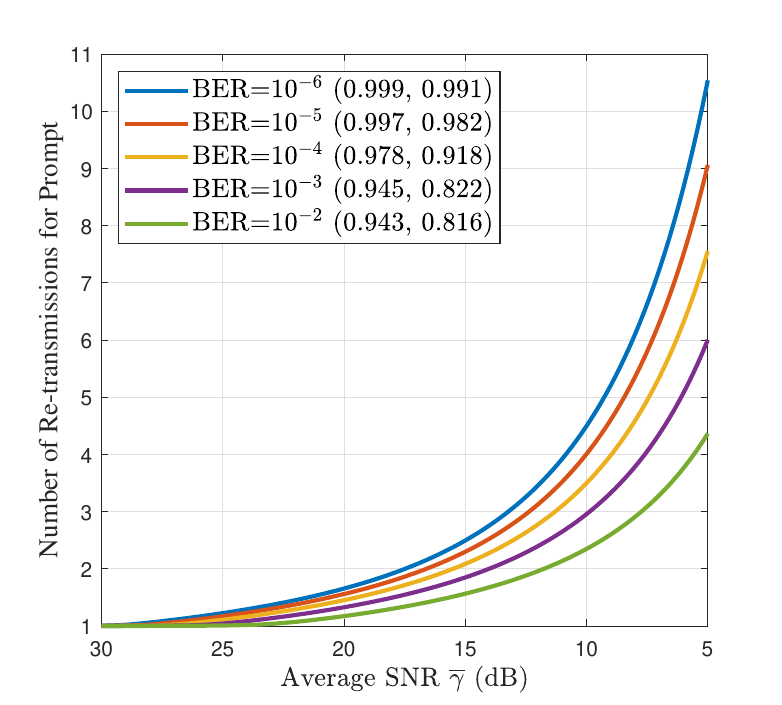}
    \caption{}
    \label{fig:sub4}
  \end{subfigure}
  \captionsetup{font={footnotesize}, singlelinecheck = off, justification = justified,name={Fig.},labelsep=period}
  \caption{Optimal wireless parameters versus average SNR $\overline{\gamma}$ at various target BERs: (a) Percentage of power for prompt; (b) Transmission Latency; (c) Modulation order (bits per symbol) for the edge map; (d) Average numbers of prompt re-transmissions. The (CLIP, MS-SSIM) notation in the legend represents the normalized CLIP/MS-SSIM values achieved for each curve.}
  \label{fig:four_subfigures}
  \vspace{-2.5mm}
\end{figure}

\vspace{-2.5mm}
\begin{table}[t]
\captionsetup{font=small}
\caption{Modulation adaptation for the proposed framework.}
\scriptsize
\centering
\begin{tabular}{|c|p{1.6cm}|p{1.6cm}|p{1.6cm}|c|}
\hline
\diagbox[dir=SE,innerwidth=1.2cm]{$M_1$}{$\overline{\gamma}$ 
($\overline{\gamma_1}$)}{Targets} & 
(0.999, 0.991) & (0.997, 0.982) & (0.978, 0.918) & $T$ (ms)\\
\hline
64 & \makecell{$30 \geq \overline{\gamma}\geq 24.2$\\($\overline{\gamma_1}=27$)} & \makecell{$30 \geq \overline{\gamma}\geq 23.4$\\($\overline{\gamma_1}=26$)} & \makecell{$30 \geq \overline{\gamma}\geq 22.2$\\($\overline{\gamma_1}=25$)} & 1.31\\
\hline
16 & \makecell{$24 \geq \overline{\gamma}\geq 18$\\($\overline{\gamma_1}=21$)} & \makecell{$23.2 \geq \overline{\gamma}\geq 17.2$\\($\overline{\gamma_1}=20$)} & \makecell{$22 \geq \overline{\gamma}\geq 16$\\($\overline{\gamma_1}=19$)} & 1.97 \\
\hline
4 & \makecell{$17.8 \geq \overline{\gamma}\geq 11.2$\\($\overline{\gamma_1}=14$)} & \makecell{$17 \geq \overline{\gamma}\geq 10.4$\\($\overline{\gamma_1}=13$)} & \makecell{$15.8 \geq \overline{\gamma}\geq 9.4$\\($\overline{\gamma_1}=12$)} & 3.93 \\
\hline
0 & \makecell{$11 \geq\overline{\gamma}$} & \makecell{$10.2 \geq\overline{\gamma}$} & \makecell{$9.2 \geq\overline{\gamma}$} & / \\
\hline
\end{tabular}
\vspace{-5mm}
\label{tab:regions}
\end{table}

\vspace{-1mm}
\subsection{Computation latency}
The computation latency depends on the choice of the AI models and the implementation scenario. For the on-device implementation scenario, smaller models like BLIP (55 GFLOPs) \cite{li2022blip} and mobile diffusion (153 GFLOPs) \cite{MobileDiff} are used. Using a state-of-the-art A17 Pro chip (35 TFLOPS/s), the resulting computation latencies are 1.6 ms for BLIP and 4.6 ms for mobile diffusion. These are comparable to the transmission latencies as reported in Fig. \ref{fig:sub2} and TABLE \ref{tab:regions}, emphasizing the need to minimize the transmission latency for reduced end-to-end latency. 


\vspace{-4.5mm}
\section{Conclusions}
This paper proposed a latency-aware and channel-adaptive semantic communication framework with pre-trained foundation generative AI models. In this framework, the transmitter extracts multi-modal semantic content of the input signal including a textual prompt and conditioning signals. The extracted semantics are then transmitted in multiple streams with semantic-aware appropriate coding and communication schemes and then input to a generative diffusion model at the receiver. Simulation results on natural images showcase the efficacy of the proposed framework in achieving ultra-low-rate, low-latency, and channel-adaptive semantic communications.


\vspace{-3mm}

\end{document}